\documentclass[peerreview, 11pt]{IEEEtran}
\usepackage[dvips]{graphicx,epsfig}
\usepackage{epsfig}
\usepackage{graphicx}
\usepackage{subfigure}

\usepackage{amsmath}
\usepackage{amssymb}
\usepackage{color}
\usepackage{colordvi}
\usepackage{float}

\begin{document}

\setlength{\leftmargini}{0\leftmargini}
\newtheorem{definition}{Definition}
\newtheorem{lemma}{Lemma}
\newtheorem{proposition}{Proposition}
\newtheorem{corollary}{Corollary}
\newtheorem{theorem}{Theorem}
\newtheorem{conjecture}{Conjecture}
\newtheorem{remark}{Remark}
\newtheorem{example}{Example}

\newcommand{\dsum}{\displaystyle\sum}
\newcommand{\dfr}{\displaystyle\frac}
\newcommand{\dint}{\displaystyle\int}
\newcommand{\dprod}{\displaystyle\prod}
\newcommand{\naturals}{\ensuremath{\mathbb{N}}}
\newcommand{\reals}{\ensuremath{\mathbb{R}}}
\newcommand{\expectation}{\ensuremath{\mathbb{E}}}
\setcounter{page}{1}

\title{Error Bounds for Repeat-Accumulate Codes Decoded via Linear Programming}


\author{Idan Goldenberg \hspace*{1cm} David Burshtein \\
School of Electrical Engineering \\
Tel-Aviv University \\
Tel-Aviv 69978, Israel \\
E-mail: {\tt \{idang,burstyn\}@eng.tau.ac.il} }

\date{}
\maketitle

\baselineskip=18pt
\begin{abstract}
\baselineskip=16pt
We examine regular and irregular repeat-accumulate (RA) codes with repetition degrees which are all even.
For these codes and with a particular choice of an interleaver, we give an upper bound on the decoding
error probability of a linear-programming based decoder which is an inverse polynomial in the block length.
Our bound is valid for any memoryless, binary-input, output-symmetric (MBIOS) channel.
This result generalizes the bound derived by Feldman et al., which was for regular RA($2$) codes.
\end{abstract}


\section{Introduction}
Since the discovery \cite{TurboCodes1993} of Turbo codes in 1993,
there has been much focus on understanding why they perform superbly as they do.
The discovery of Turbo codes also sparked
an abundance of research into LDPC codes which were originally discovered
by Gallager \cite{Gallager1963}. This vast study of Turbo and
LDPC codes, as well as their many variations, has mainly
been with respect to two types of decoders: Optimal
maximum-likelihood (ML) and sub-optimal iterative
message-passing algorithms. The latter have been extensively
researched with several variations of the decoding algorithm,
producing in some cases an accurate understanding of the decoder
performance.

Recently, a novel decoding scheme based on linear programming (LP)
was proposed. Initially, an LP-based decoder was proposed for Turbo
codes by Feldman et al \cite{FeldmanKarger2002} with an explicit
performance bound given for repeat-accumulate (RA) codes, a variant
of Turbo codes. Later, another LP-based decoder was proposed for LDPC
codes by the same authors \cite{FeldmanWK2005}. These results, among others, have been
well-summarized in \cite{FeldmanPhdThesis}.
Further results for the LP decoder of LDPC
codes include the characterization of
pseudocodewords, and in particular, minimum-weight pseudocodewords
(e.g., \cite{VTKoetter2004,SmarVT2007}); results for the binary
symmetric channel (BSC) on the error-correction capability
(e.g., \cite{Feldman2007,DaskalakisDimakis2008}), and others.

One interesting property of the LP decoder is the \emph{ML
certificate} property. That is, that whenever the LP decoder outputs
a codeword, it is guaranteed to be the ML codeword. Iterative
message-passing algorithms do not share this property. On the other
hand, iterative algorithms have in some cases the advantage of lower
decoding complexity, as compared to LP decoding. However, for LDPC
codes this advantage is all but eliminated (see
\cite{VontobelKoetter2006,Burshtein2008}).

Compared to iterative decoding, there has so far been less
research on LP-based decoding. While the first analytic result for
LP decoding \cite{FeldmanKarger2002} has been for the case of RA
codes, most of the results for LP decoding thereafter refer to LDPC
codes. In \cite{FeldmanKarger2002}, regular RA($2$) codes were
examined, based on flow theory and graph-theoretical arguments.
Halabi and Even \cite{HalabiEven2005} have proposed a better bound
on RA($2$) codes which is based on a more careful examination of the
underlying graph-theoretical nature of the problem. Irregular RA code ensembles
have been shown to achieve excellent performance under iterative message-passing
decoding. For example, in the BEC there are known capacity-achieving sequences of codes (see e.g., \cite{PSU2005}).
Motivated
by these results for iterative decoding, we examine regular and irregular
RA codes under LP decoding. We show how to extend the results of \cite{FeldmanKarger2002}
to regular RA($q$) codes for even $q$ and to irregular RA
ensembles where all repetition degrees are even. The essential novelty in this work
is the application of Euler's (graph-theoretic) theorem to an appropriately defined
(hyperpromenade) graph.

The remainder of the paper is organized as follows. Preliminary
material is given in Section~\ref{Preliminaries Section}.
Section~\ref{Regular RA bound section} contains the derivation of
our error bound for regular RA codes. A discussion of these results
as well as their extension to irregular RA codes appears in
Section~\ref{Discussion and Numerical Results section}.
Section~\ref{Summary section} concludes the paper.
\section{Preliminaries}\label{Preliminaries Section}
In this section, we give our nomenclature and some necessary
preliminary material. Our notations largely follow those of Feldman
\cite{FeldmanPhdThesis}. In the rest of the paper, we will deal
exclusively with repeat-accumulate (RA) codes. These codes were
proposed by Divsalar et al. in \cite{Divsalar98}, in which regular
code ensembles were defined, and later generalized to irregular
ensembles in \cite{Jin-brest00}. Repeat-accumulate codes feature a
simple encoder structure and are known to have good decoding
performance under iterative message-passing decoding. The encoder of
a \emph{regular} $RA(q)$ code, shown in Fig.~\ref{RA encoder
figure}, takes an input block of $k$ bits, applies a $q$-fold
repetition code to obtain a block of $n=qk$ bits, interleaves the
block and finally feeds it into a rate-1 accumulator. The
accumulator is a recursive convolutional encoder with one memory
element which outputs at time index $t$ simply the mod-$2$ sum of
the inputs up to time $t$. The code rate in this case is
$R=\frac{1}{q}$. In an \emph{irregular} code, the number of times a bit is repeated (or its
\emph{repetition degree}) is not constant. The fraction of bits which are repeated a certain
number of times by the encoder is known as the \emph{degree
distribution}, and it is usually expressed either in vector form or in
polynomial form.

\begin{figure}[h]
\begin{center}
\leavevmode
\input{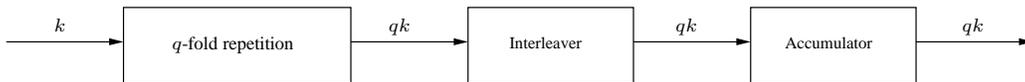}
\caption{Block diagram of a regular RA encoder.}
\label{RA encoder figure}
\end{center}
\end{figure}

Our analysis will focus on transmission of regular or irregular RA codes over memoryless, binary-input,
output-symmetric (MBIOS) channels. We denote by $x_i \in \{0,1\}, i=1,\dots,k$ the $i$'th bit of the codeword to be
transmitted and by $y_i = y(x_i)$ the channel modulation of the $i$'th bit. Since the channel is memoryless,
the $i$'th received symbol $\tilde{y}_i$ depends only on $y_i$ by the conditional
probability law $P(\tilde{y}_i|y_i)$ imposed by the channel.

The \emph{log-likelihood ratio} (LLR) $\gamma_i$ is defined as
\begin{equation}\label{Definition of LLR}
    \gamma_i = \ln \left( \frac{P(\tilde{y}_i| y_i=y(0))}{P(\tilde{y}_i| y_i=y(1))} \right)
\end{equation}
\begin{example}
In the binary symmetric channel (BSC) the channel input alphabet is binary, and so we have $y_i=x_i$.
The log-likelihood ratio is $\gamma_i = \ln \left(\frac{1-p}{p} \right)$ if $\tilde{y_i}=0$ and
$\gamma_i= \ln \left(\frac{p}{1-p} \right)$ if $\tilde{y_i}=1$.
\end{example}
\begin{example}
Consider the binary-input additive white gaussian noise (AWGN) channel.
Following conventional notation, we map bit $0$ to $+1$ and bit $1$ to $-1$, i.e., we have $y_i=1-2x_i$.
In the AWGN channel we have
\begin{equation*}
    \tilde{y}_i=y_i+z_i
\end{equation*}
where $z_i$ is a normally-distributed random variable, $z_i \sim \mathcal{N}(0,\sigma^2)$.
The LLR in the AWGN channel is easily shown to be $\gamma_i = \frac{2\tilde{y}_i}{\sigma^2}$.
\end{example}

It is convenient for purposes of analysis to rescale the LLR.
In the BSC, the rescaling enables to have $\gamma_i = 1$ if $\tilde{y_i}=0$ and $\gamma_i = -1$ if $\tilde{y_i}=1$.
In the AWGN channel, rescaling allows us to express $\gamma_i = \tilde{y}_i$.

\subsection{A Linear Program to Decode Repeat-Accumulate Codes}
We are interested in the performance of a linear-programming (LP) decoder for RA codes.
To make our presentation self-contained, we briefly present the linear program proposed
by Feldman \cite{FeldmanPhdThesis}.

First, we look at the accumulator section of the encoder, assuming at this stage that it were
the entire encoder. The accumulator is a rate-1 convolutional encoder, and has a state diagram
and trellis as shown in Fig. \ref{accumulator: trellis and state diagram}. The trellis $T$ features
connections or \emph{edges} describing transitions between states from successive time intervals, which are labeled according to the output
of the accumulator. Each edge also has a 'type' which depends upon the input bit triggering the transition.
Note that the trellis contains an extra layer used to terminate the
code. Adding the extra bit to force the encoder back to the
zero-state incurs a small loss in the code rate, but makes
analysis more convenient, since each codeword corresponds to a
"cycle" rather than an arbitrary path in $T$.
\begin{figure}[h]
  \begin{center}
    \subfigure[State diagram of an accumulator]{
    \label{RA encoder state diagram}
    \includegraphics[scale=0.65]{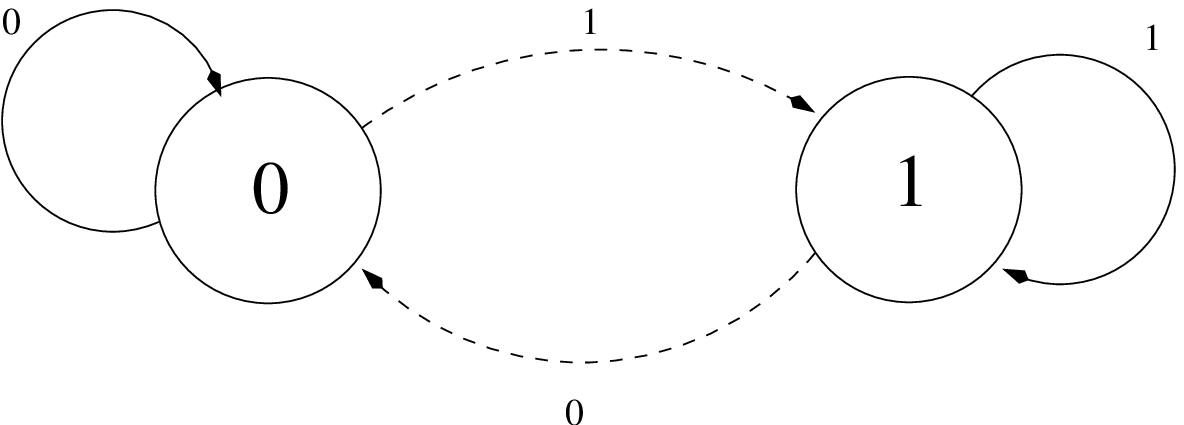}
    } \\
    \subfigure[The trellis for an accumulator]{
    \label{Accumulator trellis}
    \input{trellis.pstex_t}
    }
  \end{center}
  \caption{State diagram and trellis of an accumulator. Each state transition has
a label according to the output of the accumulator. Solid lines
denote state transitions when the input bit is zero, and dashed
lines are used when the input bit is one.}
  \label{accumulator: trellis and state diagram}
\end{figure}
All edges have a direction (i.e., forward in time).
Also, every edge $e \in T$ is assigned a cost $\gamma_e=\gamma_i$ where the index $i$ is selected
according to the trellis segment containing the edge.
Consequently, it can be shown that finding the ML codeword is equivalent to
finding the minimum-cost path traversing across the trellis. If
indeed the accumulator were the entire code, there would be no
additional constraints on the path, and finding the one with mininal
cost could be accomplished, for example, by using the Viterbi
algorithm.

We now take the effect of the (possibly irregular) repetition code and interleaver into account.
Assume that input bit $x_t$ has repetition degree $q_t$. If so, we would expect the inputs
into the accumulator at some set of indices $X_t \triangleq \{i^1,i^2,\dots,i^{q_t}\}$ to be
identical (obviously, this set depends on the interleaver).
Translating this into trellis terms, we require that for all $t=1,\dots,k$ we have the same type of edge
at all layers $i \in X_t$. Any path satisfying this requirement is called an \emph{agreeable} path.

A linear program to decode RA codes (RALP) was defined by Feldman \cite{FeldmanPhdThesis} as follows:
\begin{eqnarray}
  \text{RALP:  minimize   } & \dsum_{e \in T} \gamma_e f_e & \text{s.t.} \nonumber \\
  & \dsum_{e \in out(s_0^0)} f_e =1 & \label{force unit flow}\\
  & \dsum_{e \in out(s)} f_e = \dsum_{e \in in(s)} f_e & \forall s \in T \backslash \{s_0^0,s_n^0\} \label{flow conservation}\\
  & x_t = \dsum_{e \in I_i} f_e & \forall i \in X_t, \; t=1,\dots,k \label{agreeability constraints} \\
  & 0 \leq f_e \leq 1 & \forall e \in T \nonumber
\end{eqnarray}
where $in(s)$ is the set of edges entering node $s$, $out(s)$ is the set of edges
exiting node $s$, and $I_i= \{ (s_{i-1}^0,s_i^1),(s_{i-1}^1,s_i^0) \}$ is the pair of
"input-1" edges entering layer $i$.
Equation \eqref{force unit flow} ensures that one unit of flow is
sent across the trellis. Equation \eqref{flow conservation} enforces
flow conservation at each node, i.e., that whatever flow enters must also exit.
The agreeability constraints are
imposed by equation \eqref{agreeability constraints}. These
constraints say that a feasible flow must have, for all $X_t, \;
t=1,\dots,k$, the same amount $x_t$ of total flow on input-$1$ edges
at every segment $i \in X_t$.

In order to use RALP as a decoder, one should solve the LP problem above on the trellis with
edge costs $\gamma_e$ defined by the received vector $\tilde{y}$, thus obtaining an optimum
point $(f^*,x^*)$. If $f^*$ is integral (i.e., all values
are 0 or 1), $x^*$ is output as the decoded information word. If not,
the output is "error". We refer to this algorithm as the \emph{RALP
decoder}. It can be shown that this decoder has the \emph{ML certificate}
property: whenever it finds a codeword, it is guaranteed to be the
ML codeword.

\section{An Error Bound for Regular RA($q$) Codes with Even $q$}\label{Regular RA bound section}
In this section, we derive an upper bound on the decoding error
probability of the RALP decoder. For simplicity, we deal in this section exclusively with regular codes.
This is an extension of the results
of \cite{FeldmanPhdThesis}, which applied to RA(2) codes, to the
case of RA($q$) codes for even $q$.
For the purpose of analysis, we
define an auxiliary graph which contains subgraphs called
\emph{hyperpromenades} which carry a meaning similar to error events
in convolutional codes. The structure of these hyperpromenades
suggests a design of an interleaver. We show how to design a suitable
interleaver, and show that the RALP decoder has an
inverse-polynomial error rate (in the blocklength $n$) when this interleaver is used. Our
discussion will not depend initially on the repetition degrees being even; we will only
require this assumption later on.

Let $\Theta$ be a weighted undirected graph with $n$ vertices
$(g_1,\dots,g_n)$ connected in a line. We call these edges \emph{Hamiltonian},
as they form a Hamiltonian path along the graph. We associate a cost (weight)
$c[g_i,g_{i+1}]$ with each Hamiltonian edge $(g_i,g_{i+1})$, equal to the cost added
by decoding code bit $i$ to the opposite value of the transmitted
codeword. Formally, we have
\begin{equation}\label{Definition of edge cost}
c[g_i,g_{i+1}] = \gamma_i\left( 1-2x_i \right)
\end{equation}
where $x_i$ is the $i$'th codeword bit, and $\gamma_i$ is the
log-likelihood ratio of code bit $i$, as defined in
\eqref{Definition of LLR}. In the BSC, we have $c[g_i,g_{i+1}]=+1$
if $x_i=y_i=\tilde{y}_i$ and $c[g_i,g_{i+1}]=-1$ if $y_i \neq
\tilde{y}_i$. Naturally, the decoder does not know the costs
$c[g_i,g_{i+1}]$; they are used solely as a means for analysis. In
addition to the Hamiltonian edges described above, $\Theta$ contains
also \emph{hyperedges} connected between the vertices. A
$q$-hyperedge is an edge connecting $q$ vertices, and is formally
defined as an unordered $q$-tuple of vertices from the graph. We
connect a total of $k$ hyperedges, where hyperedge $t$
contains the vertices within the index set $X_t \; (t=1,\dots,k)$.
Note that according to this setting, exactly one hyperedge is
connected to every vertex. 
In \cite{FeldmanPhdThesis}, where the authors consider the case
$q=2$, these extra edges form a matching on the vertices of the
graph. Extending this nomenclature to any $q$, we will call them
\emph{matching} hyperedges. These edges are defined to have zero
cost in the auxiliary graph.

An \emph{atom path} $\mu(\sigma,\tau)$ is a walk which begins at
vertex $g_{\sigma}$ and finishes at vertex $g_{\tau}$, using
Hamiltonian edges only. Therefore, if $\sigma<\tau$, we have
$\mu(\sigma,\tau)=(g_{\sigma},g_{\sigma+1},\dots,g_{\tau-1},g_{\tau})$.
A \emph{hyperpromenade} $\Psi$ is a set of atom paths, possibly with
multiple copies of the same atom path in the set. The set $\Psi$ is also required
to satisfy a certain "agreeability" constraint.
Formally, define, for each segment $i$ in the trellis where $1\leq i
\leq n$, the following multiset $B_i$:
\begin{equation*}\label{Definition of B_i}
    B_i = \{\mu \in \Psi: \; \mu=\mu(\sigma,\tau),\; \text{where} \; i=\sigma \; \text{or} \; i=\tau \}
\end{equation*}
Note that if multiple copies of some $\mu(\sigma,\tau)$ exist in $\Psi$, then $B_i$ contains multiple copies as well.
We say that $\Psi$ is a hyperpromenade if, for all $t=1\dots,k$, where $X_t=\{t^1,t^2,\dots,t^q\}$, we have
\begin{equation}\label{The sets B}
    | B_{t^1} | = | B_{t^2} | = \dots =| B_{t^q} |
\end{equation}
\begin{example}\label{hyperpromenade example}
As an example, consider the auxiliary graph illustrated in
Figure~\ref{Hyperpromenade example}. In this auxiliary graph of an
RA($4$) code, the multiset
\begin{equation*}
\Psi=\{\mu(1,2),\mu(1,2),\mu(3,10),\mu(4,5),\mu(4,12),\mu(5,7),\mu(6,11),\mu(7,12),\mu(8,9),\mu(8,9)\}
\end{equation*}
is a hyperpromenade.
\end{example}

\begin{figure}[h]
\begin{center}
\input{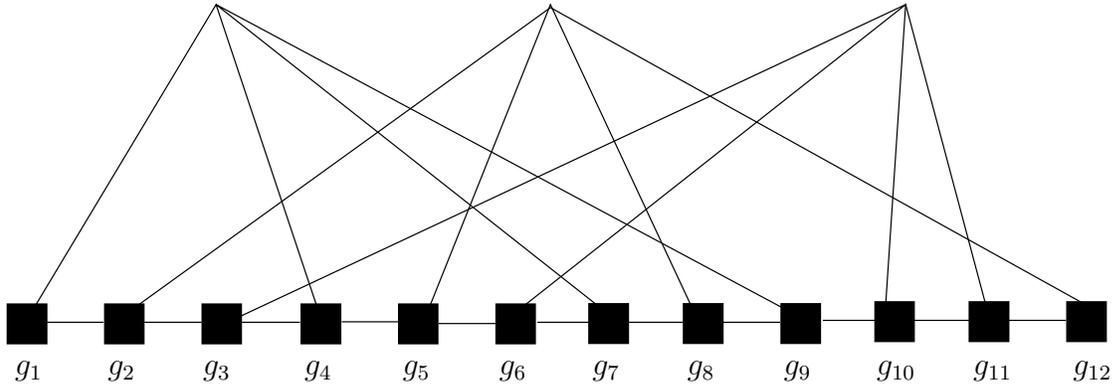}
\caption{The auxiliary graph in example \ref{hyperpromenade example}.}\label{Hyperpromenade example}
\end{center}
\end{figure}

The cost of every atom path $\mu(\sigma,\tau)$ is equal to the sum
of the costs of its edges. The cost of a hyperpromenade is equal to
the sum of the costs of the atom paths it contains, including
repeated ones.

We have the following theorem (\cite[Theorem 6.13]{FeldmanPhdThesis}).
\begin{theorem}\label{Feldmans theorem}
    For any regular RA($q$) code, the RALP decoder succeeds if all hyperpromenades have positive cost.
    The RALP decoder fails if there is a hyperpromenade with negative cost.
\end{theorem}
This theorem was stated in \cite{FeldmanPhdThesis}, and was therein proved for the case of
$q=2$. The same proof applies also for $q>2$. However, in order to use this result we need to show how to
construct graphs $\Theta$ which yield good interleavers for RA($q$)
codes; we will show that these graphs have a small probability
of having a negative-cost hyperpromenade. A key metric in our
analysis will be the girth of the auxiliary graph. As the auxiliary
graph contains hyperedges as well as regular edges (thus it is a
\emph{hypergraph}), the notion of girth needs to be extended. Define
a \emph{path} $p=(p_0,\dots,p_k)$ in the auxiliary graph to be a
series of vertices where every two consecutive vertices are
connected by an edge or a hyperedge. The only exception is that the
same (hyper)edge may not be traveled two times in a row; this means
that U-turns are not allowed, and also roundabouts within a
hyperedge (e.g., if $i_1,i_2,i_3 \in X_t$ for some $t$,
then $(g_{i_1}\rightarrow g_{i_2}\rightarrow g_{i_3})$ is not a
valid path). Aside from this restriction, a path may repeat
vertices, edges and hyperedges. Path length is measured in edges, so
the path $p=(p_0,\dots,p_k)$ has length $k$. A \emph{cycle} is a
path that begins and ends in the same vertex. The girth of a
hypergraph is thus the length of its shortest cycle. We further
define a \emph{simple path} (resp. \emph{simple cycle}) to be a path
(cycle) which does not repeat Hamiltonian edges but \emph{may} repeat
hyperedges.


Our first step is to show that an auxiliary graph $\Theta$ with high girth
can be constructed, thus implying the existence of appropriate
interleavers. For the case of $q=2$, Erd\H{o}s and Sachs \cite{ErdosSachs1963}
(see also \cite{Biggs98}) have shown a construction for such
an interleaver. The following result is an extension to $q\geq 3$ using a similar technique.
While our subsequent error bound is valid only for even $q$, this restriction need not be imposed yet.
\begin{theorem}\label{interleaver theorem}
Let $n=qk$ be the block length of a regular RA($q$) code with $q\geq 3$ and $n\geq q^4$. Then one may construct for this code
    an auxiliary graph which is a Hamiltonian line plus $k$ $q$-hyperedges which form a matching, so that the
    auxiliary graph has girth no less than $g= \lfloor \log_q n \rfloor -1$.
\end{theorem}

Proof: See appendix \ref{Proof of interleaver theorem}.

Denote the interleaver produced by this
approach by $\pi_E$. The next step is to study the auxiliary graph of an RA code which uses $\pi_E$ as an interleaver.
We focus on the underlying nature of hyperpromenades in this graph.

{\bf A study of the structure of hyperpromenades.} First, we point
out that in the case of $q=2$, it was shown by Feldman
\cite{FeldmanPhdThesis} that every hyperpromenade is equivalent to a
cycle in $\Theta$\footnote{in fact, this was the original definition
of a promenade, to which the hyperpromenade reduces for the case
$q=2$.}. This observation simplifies the analysis, as one must deal
solely with simple cycles in the auxiliary graph. In our case where
$q>2$, this is not necessarily true. Therefore, our conclusions must
be based only on the definition of a hyperpromenade.

Our goal will be to provide an upper bound on the probability that
the auxiliary graph contains a negative-cost hyperpromenade. Let
$\Psi$ be any hyperpromenade in $\Theta$. We construct a graph
called the \emph{hyperpromenade graph} $\Theta_{\Psi}$ as follows:
\begin{enumerate}
\item For every atom path $\mu(\sigma,\tau) \in \Psi$, draw in $\Theta_{\Psi}$
two vertices, labeled $\sigma$ and $\tau$, according to the endpoints of the atom path.
Connect the two vertices by an edge. If $\mu(\sigma,\tau)$ appears more than once in $\Psi$,
we will have multiple replicas of this structure, accordingly.
\item Merge all vertices with the same label into one vertex. At this stage, the graph may no longer be simple, i.e., there
may be vertex pairs connected by more than one edge.
\item Add the matching hyperedges to the graph.
\end{enumerate}
By this construction, it is obvious that one can reconstruct any
hyperpromenade given its hyperpromenade graph, i.e., there is a
$1-1$ relation between hyperpromenades and their graphs. We further
assign a cost to every edge $(\sigma,\tau)$ in $\Theta_{\Psi}$ as
follows : $c[(\sigma,\tau)]=c[\mu(\sigma,\tau)]$ where
$c[\mu(\sigma,\tau)]$ is the cost of the atom path in $\Theta$.
Hyperedges in $\Theta_{\Psi}$ are assigned zero cost. With this
definition, the total cost of the edges in $\Theta_{\Psi}$ is the
same as the cost of the hyperpromenade. We note that the
hyperpromenade graph may or may not be connected (in the sense that
there is a path between any two of its vertices). If it is, we call
the hyperpromenade \emph{connected}. This property is different from
the connectedness of the auxiliary graph, since now vertices common
to different atom paths are ignored unless they are at the
endpoints. As an example, we draw the hyperpromenade graph of $\Psi$
from Example \ref{hyperpromenade example} in Figure
\ref{hyperpromenade graph}. In this example, the hyperpromenade is
not connected and has two connected components.

Let $\Psi$ be a hyperpromenade which is not connected, and consider
the corresponding graph, $\Theta_{\Psi}$. It is easy to verify that
each of its connected components is the hyperpromenade graph of a
valid hyperpromenade. Therefore, $\Psi$ can be partitioned into
disjoint connected hyperpromenades $\Psi_1,\Psi_2,\dots,\Psi_M$
satisfying $c[\Psi]=c[\Psi_1]+\dots+c[\Psi_M]$. If $\Psi$ is a
negative-cost hyperpromenade, it thus must have a component with
negative cost.
\begin{figure}[htp]
\begin{center}
\includegraphics[scale=0.85]{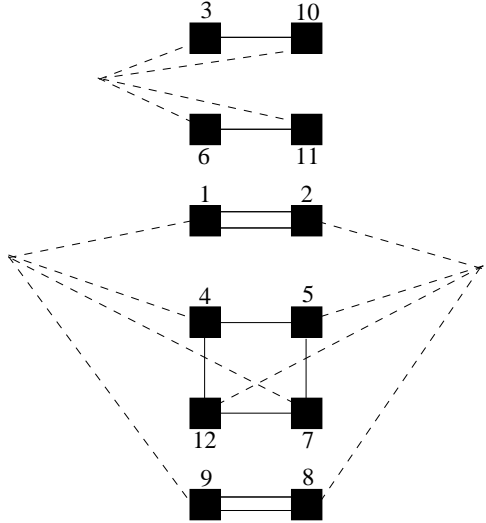}
\caption{The hyperpromenade graph $\Theta_{\Psi}$ of the hyperpromenade from
Example \ref{hyperpromenade example}.}\label{hyperpromenade graph}
\end{center}
\end{figure}
Therefore, by Theorem~\ref{Feldmans theorem}, the probability that the RALP decoder fails is the same
as the probability of having a connected hyperpromenade with
negative cost.

The next step is to establish that it is enough to look at simple
paths and cycles which are contained in a connected negative-cost
hyperpromenade. In the following, we assume $n=q^{4l+1}$ for some
integer $l$ to avoid floors and ceilings, although our arguments do
not rely on this.

\begin{theorem}\label{cycle theorem}
    Let $\Theta$ be the auxiliary graph of a regular RA($q$) code with $q$ even.
    Assume $\Theta$ has girth at least $g$, where $g\triangleq \log_q n -1$ (by
    Theorem \ref{interleaver theorem}, this girth is attainable). If
    there exists a hyperpromenade $\Psi$ in $\Theta$ with $c[\Psi] \leq 0$, then
    there exists a simple path or cycle $Y$ in $\Theta$ that contains
    $\frac{g}{2}$ Hamiltonian edges, and has cost $c[Y] \leq 0$.
\end{theorem}

\begin{proof}
Let $\Psi$ be a hyperpromenade with $c[\Psi] \leq 0$. By the
discussion above, we may assume w.l.o.g. that $\Psi$ is connected.
First, we will show that there is a cycle
$\tilde{H}=(\tilde{h}_0,\tilde{h}_1,\dots,\tilde{h}_{|\tilde{H}|}=\tilde{h}_0)$
in $\Theta$ which has $c[\tilde{H}]=c[\Psi]$, where $c[\tilde{H}]$
is measured along the edges of $\tilde{H}$. Draw the hyperpromenade
graph $\Theta_{\Psi}$, and contract the matching
hyperedges\footnote{contracting the hyperedge unites all vertices it
connects into one vertex, retaining any other edges connected to the
original vertices.}. The result is a graph, with no hyperedges, where
vertex $\sigma$ has degree $q |B_{\sigma}|$. Since $q$ is even, all
degrees are even and we can find an Eulerian tour $C$ in $\Theta_{\Psi}$, i.e., a simple cycle which passes through all
the edges. Since every edge in $\Theta_{\Psi}$ has the same cost as
its corresponding atom path in $\Psi$, we have $c[C]=c[\Psi]$. By
adding back the matching hyperedges and tracing along the atom paths
making up $\Psi$, we get from $C$ the desired cycle $\tilde{H}$ in
$\Theta$ with $c[\tilde{H}]=c[\Psi]$. Now, contract the matching hyperedges in $\tilde{H}$.
Denote by $H=(h_0,h_1,\dots,h_{|H|}=h_0)$ the contracted version of $\tilde{H}$.

No two matching hyperedges share an endpoint, and by definition the
same hyperedge cannot be used twice in a row. Therefore, at most
every other edge of $\tilde{H}$ is a matching hyperedge. Thus, and since $\tilde{H}$ is a cycle, 
\begin{equation*}
    |H| \geq \frac{1}{2} |\tilde{H}| \geq \frac{1}{2} g
\end{equation*}
Write out the cost of $H$ explicitly as
\begin{equation*}
    c[H]=\sum_{i=0}^{|H|-1} c[h_i,h_{i+1}]
\end{equation*}
Let $H_i=(h_i,\dots,h_{i+\frac{1}{2}g})$ be a subsequence of $H$
containing $g/2$ edges, and let
\begin{equation*}
    c[H_i]=\sum_{j=i}^{i+\frac{1}{2}g-1} c[h_j,h_{j+1}]
\end{equation*}
$H_i$ must be a simple path (or a simple cycle), i.e., it can have
no repeated Hamiltonian edges; otherwise, by adding the matching hyperedges back
into $H$, this would imply the existence of a cycle in $\Theta$ of
length less than $g$. Note that
\begin{equation*}
    c[\tilde{H}]=c[H]=\left(\frac{1}{\frac{1}{2}g}\right)
    \sum_{i=0}^{|H|-1} c[H_i]
\end{equation*}
since every edge is counted exactly $\frac{1}{2}g$ times. Now, if
$c[H]=c[\Psi] \leq 0$, then there must be a simple path or cycle
$H_{i^*}$ such that $c[H_{i^*}]\leq 0$. Adding back the matching
hyperedges, we get the desired simple path or cycle $Y$.
\end{proof}
Theorem \ref{cycle theorem} asserts that if there are no
negative-cost simple paths or cycles, then there is no corresponding
negative-cost hyperpromenade. It is also the first point in our
derivation which requires to use the assumption that $q$ is even. We
will now use this to get an error bound for LP decoding under the BSC.
This bound extends \cite[Theorem 6.5]{FeldmanPhdThesis}.
\begin{theorem}\label{BSC error bound}
    Consider a regular RA($q$) code ($q$ even) with block length
    $n$, and $\pi_E$ constructed in the proof of Theorem~\ref{interleaver theorem} as an interleaver.
    Assume that the code is transmitted over the BSC.
Let $\epsilon>0$ be some positive number. If the transition
probability $p$ satisfies
$p<q^{-4\left(\epsilon+1+\frac{1}{2}\log_q (4q-2)
\right)}$, then when decoded using the RALP
decoder, the code has word error probability
\begin{equation}\label{BSC bound}
    \text{WEP}<K (\log_q n) \cdot n^{-\epsilon}
\end{equation}
where $K$ is a positive constant.
\end{theorem}

\begin{proof}
By theorems \ref{Feldmans theorem} and \ref{cycle theorem}, the
decoder will succeed if all simple paths or cycles in $\Theta$ with
$\frac{1}{2}g = \frac{1}{2}(\log_q n - 1)$
(this equality is attained by definition of $\pi_E$)
Hamiltonian edges have
positive cost. We claim there are at most $n (2q-1)^{\frac{1}{2}g}$
simple paths and cycles with $\frac{1}{2}g$ Hamiltonian edges. To
see this, build a simple path or cycle by choosing any vertex
$g_{i_0}$ and traversing a simple path beginning with a Hamiltonian
edge. There are at most two choices for the first edge. If, after
traversing the Hamiltonian edge, we arrive at a vertex $g_{i_1}$,
then from $g_{i_1}$ we can choose to proceed along the second
Hamiltonian edge connected to it, or traverse a hyperedge. If a
hyperedge is traversed, there are at most $2(q-1)$ possible choices
for the next Hamiltonian edge. This gives a total of $2q-1$ choices
for the second Hamiltonian edge. Proceeding in this manner, we see
that there are no more than $(2q-1)^{\frac{1}{2}g}$ simple paths or
cycles with $\frac{1}{2}g$ Hamiltonian edges beginning from the
vertex $g_{i_0}$. Choosing an arbitrary starting vertex gives a
total of no more than $n (2q-1)^{\frac{1}{2}g}$ possible simple
paths or cycles.
In the BSC, each
Hamiltonian edge has cost $-1$ or $1$. Therefore, in any simple path
or cycle $Y$, at least half of the edges must have cost $-1$ in
order to have $c[Y]\leq 0$. Consequently, we have

\begin{equation}\label{BSC substitution}
\Pr \left( c[Y] \leq 0 \right) =
\sum_{k=\frac{1}{4}g}^{\frac{1}{2}g}\binom{\frac{1}{2}g}{k}p^{k}(1-p)^{\frac{1}{2}g-k}
\leq \frac{1}{4}g\binom{\frac{1}{2}g}{\frac{1}{4}g}p^{\frac{1}{4}g}
\end{equation}
Applying the union bound over all possible choices of $Y$ (i.e., paths with
$\frac{1}{2}g=\frac{1}{2}(\log_q n - 1)$ Hamiltonian edges), we have
\begin{eqnarray*}
  \text{WEP} &\leq& n (2q-1)^{\frac{1}{2}g} \frac{1}{4}g\binom{\frac{1}{2}g}{\frac{1}{4}g}p^{\frac{1}{4}g} \\
   &\leq& \frac{1}{4} (\log_q n) \cdot n^{1+\frac{1}{2}\log_q (2q-1)}n^{\frac{1}{2} \log_q 2}n^{\frac{1}{4}\log_q p}\cdot p^{-\frac{1}{4}}\\
   &=& K (\log_q n) \cdot n^{-\epsilon}
\end{eqnarray*}
where $K=\frac{1}{4}p^{-\frac{1}{4}}$.
\end{proof}

A similar result for the AWGN channel, which extends \cite[Theorem 6.6]{FeldmanPhdThesis}
follows.
\begin{theorem}\label{AWGN error bound}
   Consider a regular RA($q$) code ($q$ even) with block length
   $n$, and $\pi_E$ constructed in the proof of Theorem~\ref{interleaver theorem} as an interleaver.
   Assume that the code is transmitted over the binary-input AWGN.
Let $\epsilon>0$ be some positive number. If the noise variance
satisfies $\frac{1}{\sigma^2} > 4 \ln q \left( 1 +
\epsilon + \frac{1}{2} \log_q (2q-1) \right)$, then the code
has, when decoded using the RALP decoder, word error probability
\begin{equation}\label{AWGN bound}
    \text{WEP}<\tilde{K} \sqrt{\frac{1}{\log_q n -1}} n^{-\epsilon}
\end{equation}
where $\tilde{K}$ is a positive constant.
\end{theorem}

\begin{proof}
In the AWGN channel, each Hamiltonian edge in the auxiliary
graph has cost
\begin{equation*}
    c[g_i,g_{i+1}]= \gamma_i \left( 1-2x_i \right) = \gamma_i y_i = (y_i+z_i)y_i = 1 + \overline{z}_i
\end{equation*}
where $\overline{z}_i \sim \mathcal{N}\left(0, \sigma^2 \right)$. Therefore, if $Y$ is any simple
path or cycle with $\frac{1}{2}g=\frac{1}{2}(\log_q n -1)$ Hamiltonian edges, we have $c[Y]=\frac{g}{2} + Z$, where
\begin{equation*}
    Z \sim \mathcal{N}\left(0,\frac{g}{2}\sigma^2\right)
\end{equation*}

For a random variable $X \sim \mathcal{N}(0,s^2)$, we have the
following inequality: for all $x>0$,
\begin{equation}\label{gaussian inequality}
    \Pr \left(X \geq x\right) \leq \frac{s}{x\sqrt{2\pi}} e^{-\frac{x^2}{2s^2}}
\end{equation}
Using \eqref{gaussian inequality} we get
\begin{eqnarray}
  \Pr (c[Y] \leq 0) &=& \Pr (Z \geq \frac{g}{2}) \nonumber \\
   &\leq & \sqrt{\frac{\sigma^2}{\pi g}} e^{-\frac{g}{4\sigma^2}} \label{AWGN substitution}
\end{eqnarray}
As in Theorem \ref{BSC error bound}, using the union bound over all possible choices of $Y$
gives
\begin{eqnarray*}
  \text{WEP} &\leq& n (2q-1)^{\frac{1}{2}g} \cdot \Pr (c[Y] \leq 0) \\
   &\leq& n (2q-1)^{\frac{1}{2}g} \cdot \sqrt{\frac{\sigma^2}{\pi g}} e^{-\frac{g}{4\sigma^2}} \\
   &=& \sqrt{\frac{\sigma^2}{\pi (\log_q n -1)}} \cdot n^{1 + \frac{1}{2} \log_q (2q-1) -\frac{\log_q e}{4\sigma^2}} \\
   &<& \tilde{K} \sqrt{\frac{1}{\log_q n -1}} n^{-\epsilon}
\end{eqnarray*}
where $\tilde{K} \triangleq \sqrt{\frac{\sigma^2}{\pi}}$.
\end{proof}
A similar result for general MBIOS channels is given here in Theorem~\ref{MBIOS error bound}.
The proof follows the lines of the proofs of theorems \ref{BSC error bound} and \ref{AWGN error bound},
and is omitted.
\begin{theorem}\label{MBIOS error bound}
    Consider a regular RA($q$) code ($q$ even) with block length
    $n$, and $\pi_E$ constructed in the proof of Theorem~\ref{interleaver theorem} as an interleaver.
    When transmitted over an MBIOS channel and decoded using the RALP
    decoder, the code has word error probability WEP satisfying
\begin{equation}\label{MBIOS bound}
\text{WEP}< n (2q-1)^{\frac{1}{2}(\log_q n - 1)} \cdot \Pr \left(
\sum_{i=1}^{\frac{1}{2}(\log_q n - 1)} \tilde{z}_i \leq 0 \right)
\end{equation}
where $\tilde{z}_i$, $i=1,\dots,\frac{1}{2}(\log_q n - 1)$ are i.i.d. random variables
with cumulative distribution function
\begin{equation*}
    \Pr \left( \tilde{z}_i \leq z \right) = \Pr \left( \gamma_i \leq z | x_i = 0\right)
\end{equation*}

\end{theorem}

We note that The BSC and AWGN error bounds in \eqref{BSC bound} and \eqref{AWGN bound} can be derived
as special cases of \eqref{MBIOS bound}.

\section{Discussion and Numerical Results}\label{Discussion and Numerical Results section}
In the last section, we have given explicit bounds on the decoding error
probability for the RALP decoder. While these bounds apply to
regular RA($q$) codes with even $q$, it is possible to extend the
error bounds to irregular RA codes, where all repetition degrees are
even. This is apparent if we note that the proofs of Theorems
\ref{Feldmans theorem} and \ref{cycle theorem}--\ref{MBIOS error
bound} do not make any assumption on the regularity of the code.
The fact that all repetition degrees must be even is required in the
proof of Theorem \ref{cycle theorem} (this, both for the regular and
irregular cases). However, we are unable to provide an extension of
Theorem \ref{interleaver theorem} to irregular codes. In other
words, the construction of an interleaver which yields an auxiliary
graph with girth $g=\log_q n -1$ does not easily extend to the
irregular case. Still, whatever girth may be achieved for a specific
auxiliary graph of an irregular RA code can be used to apply Theorem \ref{MBIOS error
bound} to any MBIOS channel. That is, if a girth $g'$ can be
achieved for an irregular graph (rather than $g=\log_q n -1$ when
$\pi_E$ is used as an interleaver), we would have that
\begin{equation}\label{Irregular MBIOS bound}
\text{WEP}< n (2q_{\text{max}}-1)^{\frac{1}{2}g'} \cdot \Pr \left(
\sum_{i=1}^{\frac{1}{2}g'} \tilde{z}_i \leq 0 \right)
\end{equation}
instead of \eqref{MBIOS bound}, where $q_{\text{max}}$ now denotes the \emph{maximum} repetition
degree; the quantity $n (2q_{\text{max}}-1)^{\frac{1}{2}g'}$ is a
revised bound on the number of possible simple paths and cycles with
$\frac{1}{2}g'$ Hamiltonian edges in the irregular graph. It
particularizes to the expression in the proof of Theorem \ref{BSC error bound} in the special case of a regular code.

\begin{table}
\begin{center}
\begin{tabular}{|c|c|}
\hline
$q$ & Threshold \\ \hline
$4$ & $2 \cdot 10^{-5}$ \\
$6$ & $1.6 \cdot 10^{-6}$ \\
$8$ & $2.7 \cdot 10^{-7}$ \\ \hline
\end{tabular}
\end{center}
\caption{BSC transition probability thresholds ensuring vanishing error probability, as derived from Theorem~\ref{BSC error bound}.}
\label{BSC numerical results}
\end{table}

In Table~\ref{BSC numerical results} we give the thresholds for the BSC transition probability for which the error bound
\eqref{BSC bound} decays to zero, for some choices of $q$. It is apparent
that the threshold worsens as $q$ increases. This is in contrast to our expectation that coding performance
should improve with the reduction of coding rate.
Obviously, one cause for this is that having a negative cost path or cycle with $\frac{1}{2}g$ Hamiltonian edges
is only a \emph{necessary} condition for decoding failure. It would be reasonable to conjecture that, as $q$
increases, many structural restrictions other than this condition must exist in order to have a negative-cost hyperpromenade.
Also, the reliance on a union bound over all possible simple paths and cycles undermines the tightness of the bound.
%
%
Furthermore, if we were to examine
an irregular RA code, it can be seen that the bound in \eqref{Irregular MBIOS bound} would
yield the same result for the irregular code as well as for a regular code with repetition
degree $q_{\text{max}}$. Thus the possible improvement (increase) in the coding rate obtained
by reducing the repetition degrees of some information symbols is not reflected in our bound.

We further note that the improvement over Feldman's work presented in \cite{HalabiEven2005} for
$q=2$ does not seem to extend to $q > 2$. This is due to the following. In \cite{HalabiEven2005},
an improvement for $q=2$ is obtained by a careful characterization of cycles
in the auxiliary graph; since for $q=2$ every cycle is a promenade, (thus finding a
cycle is a \emph{sufficient} condition for identifying a promenade) the study in \cite{HalabiEven2005}
successfully captures all error events necessary for an upper bound.
The distinction between the $q=2$ case and $q>2$ is that
in the latter case, not every cycle is a valid hyperpromenade.
Consequently, analysis of cycles alone is insufficient in order to obtain an upper bound on the
decoding error probability for $q>2$.

\section{Summary}\label{Summary section}
We have presented an upper bound on the word error probability of regular and irregular RA codes
transmitted over MBIOS channels and decoded by the RALP decoder. This bounding technique
extends the one presented by Feldman \cite{FeldmanPhdThesis} for regular RA($2$) codes to the case of
regular RA($q$) codes and to irregular RA codes with even repetition degrees.
Our technique essentially relies on applying Euler's graph-theoretic theorem to an appropriately-defined
graph (i.e., the hyperpromenade graph).

\appendices

\section{Proof of Theorem \ref{interleaver theorem}} \label{Proof of interleaver theorem}

\emph{Theorem \ref{interleaver theorem}}:
Let $n=qk$ be the block length of a regular RA($q$) code, $q\geq 3$ and $n\geq q^4$. Then one may construct for this code
an auxiliary graph which is a Hamiltonian line plus $k$ $q$-hyperedges which form a matching, so that the
auxiliary graph has girth no less than $g = \lfloor \log_q n \rfloor -1$.

In the proof we will assume $n$ is a power of $q$ to avoid using floor and ceiling notations. The
proof easily extends to the general case.

\begin{proof}
Let $H$ be a Hamiltonian cycle with $n=q^{g+1}$ vertices. Let $E_0$
be the set of edges in $H$, and $V$ the set of vertices (this is denoted by $H=(V,E_0)$).
Let $D$ denote the set of all possible $q$-hyperedges, and let $A \subseteq D$ satisfy the
following conditions
\begin{enumerate}
  \item No vertex is incident with more than one $q$-hyperedge in
  $A$. \label{Condition1}
  \item The girth of $H_A=(V,E_0 \bigcup A)$ is not less than
  $g$.\label{Condition2}
\end{enumerate}

Then we shall show that if $|A|<q^g$, there exists $A^+ \subseteq D$
such that $\left|A^+ \right| = |A|+1$ and $A^+$ satisfies
\ref{Condition1}) and \ref{Condition2}). By repeatedly applying this result
we obtain that there exists some set $A$ with $|A|=q^g$ satisfying the above two conditions.

Let $d_A$ be the distance function in $H_A$. Let $V_2(A) \subseteq
V$ denote the set of vertices with degree 2 in $H_A$, i.e., those
which are not incident with any $q$-hyperedge in $A$. Given that
$|A|<q^g$, it follows that $V_2(A)$ has at least $q$ members. If
some set of $q$ vertices $p_1,p_2,\dots,p_q \in V_2(A)$ is such that
$d_A(p_i,p_j)\geq g-1$ for all $i,j \in \{1,\dots,q\}$, $i\neq j$,
then the set $A^+=A\bigcup \left\{p_1 p_2 \dots p_q\right\}$
satisfies the required conditions.

Suppose there is no such set of $q$ vertices. Define $t$ to be the
maximum number of vertices $p_1,\dots,p_t \in V_2(A)$ such that
$d_A(p_i,p_j)\geq g-1$ for all $i,j \in \{1,\dots,t\}$, $i\neq j$.
By our assumption we have that $1 \leq t \leq q-1$. Select vertices
$p_1,\dots,p_t$ which achieve this maximum. Let
\begin{equation}\label{Dr(z)}
    D_r(z) =\left\{ v \in V | d_A(z,v) \leq r \right\}
\end{equation}

We claim that
\begin{equation}\label{v2A contained in U}
    V_2(A) \subseteq U' \triangleq D_{g-1}(p_1) \cup D_{g-1}(p_2) \cup \dots D_{g-1}(p_t)
\end{equation}
This is easily seen, as follows. Suppose there is a vertex $p_{t+1}
\in V_2(A) \backslash U'$. Then the set $p_1,\dots,p_t,p_{t+1}$ is
such that $d_A(p_i,p_j)\geq g-1$ for all $i,j \in \{1,\dots,t+1\}$;
this is in contradiction with the definition of $t$.

Set $p_1,\dots,p_t$ according to the definition above, and choose
$p_{t+1},\dots,p_q \in V_2(A)$ arbitrarily. For any $x \in V_2(A)$
the set $D_{g-1}(x)$ has size at most
\begin{equation}\label{neighborhood bound}
    1+2+2q+2q^2+\dots+2q^{g-2}=1+2\frac{q^{g-1}-1}{q-1}
\end{equation}
Consequently, if $U\triangleq D_{g-1}(p_1)\bigcup
D_{g-1}(p_2)\bigcup \dots \bigcup D_{g-1}(p_q)$, then
\begin{equation}\label{U inequality}
    \left|U\right| \leq
    \left|D_{g-1}(p_1)\right|+\left|D_{g-1}(p_2)\right|+\dots+\left|D_{g-1}(p_q)\right|
    \leq q+2q\frac{q^{g-1}-1}{q-1}
\end{equation}
Let $W=V\backslash U$. Since $|V|=q^{g+1}$, it follows from the
preceding inequality that
\begin{equation}\label{inequality of W}
    \left| W \right| \geq q^{g+1} - q -2q\frac{q^{g-1}-1}{q-1}
\end{equation}
Let $\tilde{p}_1,\tilde{p}_2,\dots,\tilde{p}_q \in W$ be arbitrary
vertices, and let $\tilde{U}\triangleq D_{g-1}(\tilde{p}_1)\bigcup
D_{g-1}(\tilde{p}_2)\bigcup \dots \bigcup D_{g-1}(\tilde{p}_q)$.
This situation is depicted in Figure~\ref{circles figure}.
\begin{figure}[h]
\begin{center}
\input{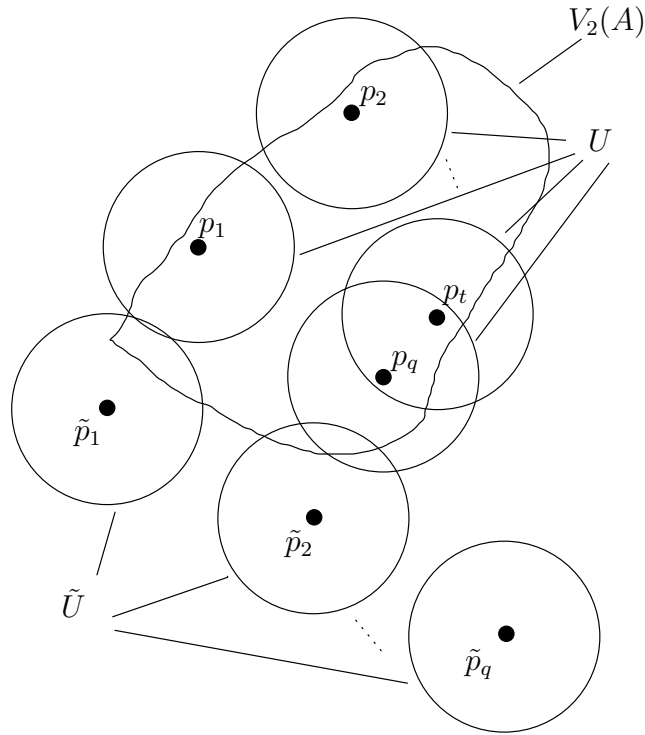}
\caption{A visualization of the some of the relationships between the sets defined in the proof.
The sphere around each noted vertex
depicts the set $D_{g-1}(\cdot)$ of vertices in its vicinity.}
\label{circles figure}
\end{center}
\end{figure}

We have that $D_{g-1}(\tilde{p}_1)$ has size at most
\begin{equation}\label{neighborhood bound, not in V2}
    1+(1+q)+(1+q)q+(1+q)q^2+\dots+(1+q)q^{g-2}=1+(1+q)\frac{q^{g-1}-1}{q-1}
\end{equation}
and therefore $|\tilde{U}|$ satisfies
\begin{equation}\label{tildeU inequality}
\left|\tilde{U}\right| \leq q+q(1+q)\frac{q^{g-1}-1}{q-1} \stackrel{(a)}{<} \left|W\right|
\end{equation}
where (a) stems from plugging in the expression from
\eqref{inequality of W} and applying some algebra \footnote{one
needs to assume here that $q\geq 3$ and $g \geq 3$; $g \geq 3$ follows from the assumption that $n \geq q^4$.}. Now,
\eqref{tildeU inequality} implies that there exist vertices
$s_1,s_2,\dots,s_q \in W$ such that any pair $i\neq j$ has
$d_A(s_i,s_j)\geq g-1$. To see this, note that one may select $s_1,s_2,\dots,s_q \in W$
sequentially, as follows: first select $s_1 \in W$ arbitrarily, then select for $i=2,\dots,q$
\begin{equation}\label{directions for selection of si}
s_i \in W \backslash \bigcup_{j=1}^i D_{g-1}(s_j)
\end{equation}
arbitrarily; \eqref{tildeU inequality} ensures that the set in the RHS of \eqref{directions for selection of si}
is nonempty. We further have by definition of $W$ that none of the vertices $s_1,s_2,\dots,s_q$ are in $V_2(A)$.
Therefore, for every $s_i, i=1,\dots,q$ there is a distinct
hyperedge $s_i s_i^{(2)}\dots s_i^{(q)}$. Consider these
$q$-hyperedges, $s_1 s_1^{(2)}\dots s_1^{(q)},s_2 s_2^{(2)}\dots
s_2^{(q)},\dots,s_q s_q^{(2)}\dots s_q^{(q)}$. Since all vertices in
$W$ have distance at least $g$ from $p_1,p_2,\dots,p_q$, it follows
that $s_j^{(i)}$, $2 \leq i \leq q$, $1 \leq j \leq q$ all have
distance at least $g-1$ from $p_1,p_2,\dots,p_q$. Therefore, the set
\begin{eqnarray*}
    A^+ &=& A \bigcup \left\{p_1 s_1^{(2)}\dots s_1^{(q)},p_2 s_2^{(2)}\dots s_2^{(q)},p_q s_q^{(2)}\dots s_q^{(q)},s_1 s_2 \dots,s_q \right\}
    \nonumber \\ && \hspace*{1cm} \backslash \left\{s_1 s_1^{(2)}\dots s_1^{(q)},s_2 s_2^{(2)}\dots s_2^{(q)},\dots,s_q s_q^{(2)}\dots s_q^{(q)} \right\}
\end{eqnarray*}
satisfies the required conditions.

Once we have built a circle of vertices and a $q$-fold matching between them,
the theorem follows by removing one of the edges along the circle; this is the nonexistent edge
in the graph $\Theta$ between the first and last nodes. Removing this edge does not reduce the
girth of the auxiliary graph, and completes the desired construction.
\end{proof}

\emph{Discussion.} The proof of the theorem incurs bounding the size
of the neighborhood sets $D_{g-1}(\cdot)$. We could have tightened
the bounds we used, e.g. in \eqref{neighborhood bound}, by noting
that if a matching edge is traversed, the next level neighbor can be
only one of two choices (a Hamiltonian edge must be used next). This
would have replaced equations \eqref{neighborhood bound} and
\eqref{neighborhood bound, not in V2} with more elaborate, albeit
precise expressions. Consequently, the girth bound would have
improved by a constant factor at most. This is of lesser importance
than the ultimate behavior of the girth of the graph which is
logarithmic in the block length. We thus omit this refinement.

Our proof uses a construction to show it is possible to build the desired graph.
We note that this construction contains degrees of freedom which can lead to different results.

\emph{The complexity of the proposed construction}. We will show
that the complexity of constructing the matching above is polynomial
in $n$, and in particular that it is no more than $O(n^{q+2})$ in
time and space. To show this, we go over the stages of the
construction and bound their complexity (some technical details are
omitted).
\begin{enumerate}
\item The basic iteration step in the construction involves adding a
hyperedge to the graph. This step is performed $k=n/q$ times. We
therefore examine the worst-case complexity of this basic step and multiply the
result by $k$.
\item In each iteration, construct for every vertex $x \in V$ the set
$D_{g-1}(x)$, in the form of a list. This entails a complexity of no
more than $n \cdot q^g=O(n^2)$, since $q^g$ is an upper bound on the size
of $D_{g-1}(x)$ (see Eq. \eqref{neighborhood bound, not in V2}).
\item Using the lists constructed in step (2), we need to determine if there exists a set of
vertices $p_1,p_2,\dots,p_q \in V_2(A)$ such that
$d_A(p_i,p_j)\geq g-1$, $i \neq j$. It can be seen that the
complexity of this step is no more than $O(n^{q+1})$. This bound
includes the complexity associated with finding the vertices
$p_1,\dots,p_t$ defined in the construction. If $t=q$, the iteration
step ends. If not, we need to construct the set $W$ and find the
vertices $s_1,s_2,\dots,s_q \in W$ such that any pair $i\neq j$ has
$d_A(s_i,s_j)\geq g-1$.
\item The construction of the set $W$ makes use of the precalculated
lists of neighbors, and can be seen to have complexity no more than
$O(n)$.
\item The final step requires finding vertices $s_1,s_2,\dots,s_q \in W$ such that any pair $i\neq j$ has
$d_A(s_i,s_j)\geq g-1$. Using the sequential construction described above,
the worst-case complexity can be seen to be $O(n^2)$. Once the vertices are found, the matching
hyperedges are added and the iteration step ends.
\end{enumerate}
The total complexity of constructing a high-girth interleaver is
thus no more than
\begin{equation*}
    \frac{n}{q} \left(O(n^2)+ O(n^{q+1}) + O(n) + O(n^2) \right) = O(n^{q+2})
\end{equation*}

\end{document}